\documentclass{article}[15pt]      % Specifies the document class

\usepackage{amsmath,amssymb,epsfig,color,calc,rotating,graphics}
\usepackage{ulem}
\usepackage{lineno}

\usepackage{float}

\usepackage{natbib}
\bibpunct{(}{)}{;}{a}{}{,}

\newtheorem{theorem}{Theorem}[section]

\newtheorem{proposition}[theorem]{Proposition}
\newtheorem{corollary}[theorem]{Corollary}
\newtheorem{conjecture}[theorem]{Conjecture}

\newenvironment{proof}[1][Proof]{\begin{trivlist}
\item[\hskip \labelsep {\bfseries #1}]}{\end{trivlist}}
\newenvironment{definition}[1][Definition]{\begin{trivlist}
\item[\hskip \labelsep {\bfseries #1}]}{\end{trivlist}}

\newcommand{\qed}{\nobreak \ifvmode \relax \else
      \ifdim\lastskip<1.5em \hskip-\lastskip
      \hskip1.5em plus0em minus0.5em \fi \nobreak
      \vrule height0.75em width0.5em depth0.25em\fi}

\begin{document}

\title{Extending 1089 attractor to any number of digits \\
and any number of steps 
\vspace{0.2in}
\author{
Yannis Almirantis$^1$ and
Wentian Li$^{2,3}$ \\
{\small  1. Theoretical Biology and Computational Genomics Laboratory, Institute of Bioscience and Applications}\\
{\small National Center for Scientific Research ``Demokritos", Athens, Greece}\\
{\small  2. Department of Applied Mathematics and Statistics, Stony Brook University, Stony Brook, NY, USA}\\
{\small  3. The Robert S. Boas Center for Genomics and Human Genetics}\\
{\small  The Feinstein Institutes for Medical Research, Northwell Health, Manhasset, NY, USA}\\
\\
}
\date{}
}  % Declares the document's title.
\maketitle                   % Produces the title.
\markboth{\sl }{\sl }

\normalsize

\vspace{-0.2in}

\begin{center}
{\bf ABSTRACT}
\end{center}

The well-known ``1089" trick reflects an amazing trait of digital reversal process
and reminisces of a limiting attractor in dynamical systems even though
it takes only two steps. It is natural to consider the situations when
the number of digits is beyond three as in the original 1089 trick, 
as well as situations when the number of steps is beyond two.
The generalization to a larger number of digits has been mostly done
by Webster's work which we will reproduce. After the two steps
of the ``1089" trick  for any number of digits,
the resulting integers, the number of which is very low,
are named here ``Papadakis-Webster integers" (PWI).
A PWI is always divisible by 99, and the resulting quotients consist of only
0's and 1's, which we name ``Papadakis-Webster binary strings" (PWBS).
Not all binary strings could be PWBS, and we 
define the ``hairpin pairing rule"
to determine if a binary string is a PWBS. To generalize 1089 trick
to any number of steps, we propose a two-option iteration procedure
named ``iterative digital reversal" (IDR) suitably interweaving additions and subtractions. 
The simplest limiting behavior of IDR
is 2-cycles. The elements in an IDR 2-cycle are all 
composed of repetitions of the 10(9)$_L$89 
($L \ge 0$) motif, and are all PWIs. The lower 2-cycle
elements after division of 99 belong to the subset of PWBS that 
are palindromic and consist of
0- and 1-blocks with a minimal length of two. IDR also has longer  $p$-cycles
($p=$ 10, 12, 71) whose elements seem to contain at least one PWI. 
Another interesting finding about IDR is that it contains non-periodic
and diverging trajectories, as the integer values grow to infinity.
In these diverging trajectories, while the number of flanking digits around the middle 
point increases by the iteration, the middle part
has an 8-cycle rhythm or signature which has been found in all diverging trajectories.
Overall, the generalization of the original 1089 trick in both ``space"
and ``time" leads to many new patterns in integers and new phenomenology in dynamics.

\newpage

\large

\section{Introduction}

\indent

Prof.Acheson recalled reading about the ``1089 trick" \citep{ball} as a 10 years old
from the I-SPY magazine in the 1950s' England \citep{acheson}: think of any integer
with three digits, making sure the first and the last digits differ by at least two
(e.g., 782); reverse the digits and take the absolute difference (e.g., 782-287=495); 
reverse that new integer and take the sum (495+594); the end result is always 1089.
What causes this uniqueness in the final answer, or in a 
physicist's jargon, the universality, regardless of the differences
in the initial integer  \citep{stanley,deift}?

For those who have heard of the 1089 trick, there 
are new questions raised. Among them: 
what about an initial integer with 4-digit, 5-digit, or any number of digits -- 
will the end result of the above two operations still be unique?  What are 
the other integers besides 1089 that behave like 1089?  What if the first 
digit and the last digit are equal or only differ by one?  What if 
the two operations are extended to any number of operations?

In this paper, we try to provide a relatively complete set of answers to these questions.
The extension of 1089 trick to arbitrary number of digits was
essentially answered by Roger Webster, in a probably less known paper \citep{webster}.

In an even earlier work, Constantinos Papadakis, a Greek engineer and inventor, also 
attempted to generalize the 1089 trick to any number of digits \citep{papadakis}.
He described the properties of what we call here Papadakis-Webster numbers and binary 
strings and emphasized the significance of their very low populations.  
Still, his proofs were somewhat amateurish and incomplete. His work is even less known 
because his monograph is not in English.
For historical reasons, we  made
his monograph freely available on the Internet, with an introduction in English,
at 
https://shorturl.at/ILGdC.

Webster's and Papadakis's results will be reproduced here.  In order to extend the 1089 trick
to any number of steps, we first need to design a rule concerning when to
use subtraction and when to use addition. Our rule is the following: if
the reversed integer is smaller, subtracting the two; if the reversed integer
is larger (or equal), adding the two.

An operation that is repeated an infinite number of steps can be considered
as a dynamical system. For the 3-digit situation discussed above, the step
after reaching 1089 is an addition according to our rule, because the reversed
integer is larger, then next integer is 1089+9801=10890. The step after reaching 10890 is
a subtraction because the reversed integer is smaller, then the next integer
is 10890-9801= 1089, back to the original 1089.  No matter how many steps we would take further,
the system is settled on the 1089-10890 2-cycle ``attractor".  All 3-digit 
integers where the first digit is larger than or equal to the last digit 
by 2 are  part of  ``basin of attraction" of the 1089-10890 attractor.  

Because all 3-digit integers are in the basin of attraction -- none are 
attractor themselves (1089 and 10890 are 4 and 5-digit integers),
for 1089 trick,
the state space of all integers collapse to a very small subset by our operation.
There is a strong constraint on what the integers in the attractor
would look like. Two important constraints on the integers after two steps
were discovered by Webster: (1) these are divisible by 99; (2) after divided
by 99, the quotient is a binary string (consisting of 0s and 1s only).
In fact, this binary string is related to the digital-borrowing sequence
during the first subtraction step. 
We will examine to what degree these are still true
when we extend the 1089 trick to any number of steps.

To our surprise, 2-cycle attractors are not the only possible end results
of our operation with infinite number of steps. 
Other longer cycles are possible. Also, there are
even non-periodic (acyclic) behaviors:
the integers become larger and larger as we continue to iterate
the mapping, even though subtraction part of the operation is always used. 
There are constraints on the integers in the diverging trajectories.
There are also 8-cycle rhythm or ``signature" in 
the digital sequence of these non-attractor integers.
All these results are well beyond the original 1089 trick and Webster's work.
 
The paper is organized as follows: 
section 2 is on extending 1089 trick 
to larger integers with more than 3 digits, principally studying
the particular properties of the (very few) end-integers of the procedure. 
It has two subsections.
One is about the 
produced end-integers themselves, the Papadakis-Webster integers.  The 
other is about the Papadakis-Webster integers (PWI) divided by 99, which always
end up to integers that consist of 0's and 1's only, called Papadakis-Webster
binary strings (PWBS). 
Section 3 is about extending 1089 trick to an infinite number of
steps which contains five subsections. The first subsection introduces the
iterative digital reversal (IDR) mapping and the general description of
its dynamical behavior. The second and the third subsections describe 
the 2-cycle attractors of IDR and its connection to the Papadakis-Webster integers.
The fourth and fifth subsections are about higher cyclic attractors, and
acyclic diverging trajectories. The paper ends with the Discussion section. 
The Appendices contain discussion on a caveat in the 1089 trick,
on the requirement that the number of digits after subtraction has to remain
the same; more examples of PWI not listed in the main text; an example
of a 71-cycle; and a short discussion on the 1089 trick and IDR beyond
decimal numerical system. 

\section{Extension of 1089 trick to integers with any number of digits}

\subsection{Introducing the Papadakis-Webster integers and Papadakis-Webster
binary strings}

\indent

\begin{definition}
{\bf (Digital reversal):}
For any integers with $n+1$ digits, $D=\sum_{i=0}^n a_i 10^i=(a_na_{n-1}\cdots a_2a_1a_0)$,
its digital reversal $rev(D)$ is defined as $rev(D)\equiv \sum_{i=0}^n a_{n-i} 10^i 
=(a_0a_1a_2 \cdots a_{n-1}a_n)$.
\end{definition}

\begin{definition}
\label{def-PWI}
{\bf (Papadakis-Webster integers (PWI) ):} For a (n+1)-digit integer $D$, assume 
(1) $D > rev(D)$, (2) $D$ and $D-rev(D)$ have
the same number of digits;  then the following two steps are carried out: 
(1) $E=D-rev(D)$, and (2) $F=E+rev(E)$; 
The end result F is defined as a Papadakis-Webster integer (PWI). 
\end{definition}
Note: $a_n >   a_0+1$ is a sufficient, but not necessary, condition 
for $D$ and $E=D-rev(D)$ having the same number of digits.
When $a_n=a_0+1$, $E$ may or may not have the same number of digits as
$D$, and $F$ may or may not be a PWI.

\begin{definition}
{\bf (Digital borrow and carryover sequence):} In subtracting $rev(D)$ from $D$,
the digital borrow sequence $\{ b_i \}$ is defined as the binary indicator:
$b_i =1$ if the subtraction at position $i$
borrows from position $i+1$, and $b_i=0$ if not.
Similarly, in adding $D$ and $rev(D)$, the digital carryover sequence $\{ c_i \}$
is defined by $c_i$ if the summation at position $i$ is larger than 10,
and $c_i=0$ if not. The $\{ b_i \}$ and $\{ c_i\}$ sequences satisfy these
relations:
\begin{equation}
\label{eq-borrow}
b_i = \left\{ 
\begin{matrix}
1 & \mbox{if $a_i-a_{n-i} -b_{i-1} <  0$} \\
0 & \mbox{if  $a_i-a_{n-i} -b_{i-1} \ge 0$ } \\
\end{matrix}
\right.
\end{equation}
\begin{equation}
\label{eq-carry}
c_i = \left\{ 
\begin{matrix}
1 & \mbox{if $a_i+a_{n-i} +c_{i-1} \ge 10$} \\
0 & \mbox{if  $a_i+a_{n-i} +c_{i-1} < 10$ } \\
\end{matrix}
\right.
\end{equation} 
\end{definition}
Usually $b_i$ and $b_{n-i}$ have different values expect for some special
situations.  Similarly, $c_i$ and $c_{n-i}$ usually should have the same 
value unless $a_i+a_{n-i}=9$.  

\begin{proposition}
\label{prop-borrow}
A Papadakis-Webster integer only contains the digit-borrowing information
during the $D-rev(D)$ step, and does not contain information about the original
digits $\{ a_i \}$.
\end{proposition}

\begin{proof}
It is easy to check that the following formula:
\begin{eqnarray}
\label{eq-E}
E &=& D-rev(D) \nonumber \\
 &=& \sum_{i=0}^n a_i 10^i - \sum_{i=0}^n a_{n-i} 10^i \nonumber \\
 &=& \sum_{i=0}^n (a_i - a_{n-i} +10 {b_i} - b_{(i-1)}) 10^i 
\end{eqnarray}
contains the terms that correspond to the necessary borrowing
operations for all digits for $D-rev(D)$, 
as a borrowing at position $i$ will increase the value at position $i$ by 10,
and at the same time, a borrowing at position $i-1$ will decrease the value at position $i$
by 1. We define $b_{(-1)}=0$. Also $b_n=0$ is always true because our assumption that $a_n > a_0+1$.
The same assumption also ensures that the leading digit of E can not be zero. In other
words, $E$ has the same length n+1 as $D$.

Next,
\begin{eqnarray}
\label{eq-F}
F &=& E+rev(E) \nonumber \\
 &=& \sum_{i=0}^n (a_i - a_{n-i} +10 b_{i} - b_{(i-1)}) 10^i
+ \sum_{i=0}^n (a_{n-i} - a_{i} +10 b_{(n-i)} - b_{(n-i-1)}) 10^i \nonumber \\
&=& \sum_{i=0}^n (10 b_{i} - b_{(i-1)} + 10 b_{(n-i)} - b_{(n-i-1)}) 10^i 
\end{eqnarray}
Since F does not contain information on $\{ a_i \}$, but only information in digit-borrowing
during $D-rev(D)$, proposition \ref{prop-borrow} has been proven.
$\blacksquare$
\end{proof}

Proposition \ref{prop-borrow} explains why information concerning the original
digits $\{ a_i \}$ has been lost, and only partial information
on which digit is larger than which other is kept. This great reduction on
the detailed information is the basis for universality in 1089 trick.

\begin{theorem}
\label{theorem-99}
A Papadakis-Webster integer (PWI) is divisible by 99, and the quotient
is a binary string (digits can only be 0 and 1s). 
\end{theorem}

\begin{proof}
After summation, the first two terms in Eq.(\ref{eq-F}) 
result to a complete mutual annihilation for all indices
(note the re-indexing in the summation, and  $b_{(-1)}=0, b_n=0$) :
\begin{equation}
\sum_{j=1}^{n+1} 10 b_{(j-1)} 10^{j-1} - \sum_{i=0}^n  b_{(i-1)} 10^i
= \sum_{j=1}^{n+1} b_{(j-1)} 10^j - \sum_{i=0}^n  b_{(i-1)} 10^i=0 \nonumber
\end{equation}
Then, the next two terms in  Eq.(\ref{eq-F}) can be rewritten as
(again, note the re-indexing in the summation):
\begin{eqnarray}
\label{eq-binary}
F &=& \sum_{j=-1}^{n-1} 10 b_{(n-j-1)} 10^{j+1} - \sum_{i=0}^n b_{(n-i-1)} 10^i \nonumber \\
& =&  \sum_{j=-1}^{n-1} 100 \cdot b_{(n-j-1)} 10^{j} - \sum_{i=0}^n b_{(n-i-1)} 10^i \nonumber \\
&=& 99 \cdot \sum_{i=0}^{n-1} b_{(n-i-1)} 10^i \blacksquare
\end{eqnarray}
\end{proof}

\begin{corollary}
A Papadakis-Webster integer consists of only digits 0, 1, 8, and 9s.
\end{corollary}
As can be seen from Eq.\ref{eq-binary}, a Papadakis-Webster integer consists of adding
99s shifted by arbitrary number of digit positions. If two 99s are shifted by one position,
their sum is 99+990=1089; if shifted by two positions 99+9900=9999; if shifted by more
than $k \ge 2$ positions, the sum is 99$(0)_{k-2}$99. 
In the case of adding (e.g.) three shifted 99s, (e.g.) 
99+990+9900=1089 +9900= 10989, no other types of digits are created.  Combining 
all these possibilities, a Papadakis-Webster integer only consists of 0, 1, 8 and 9s.

\begin{definition}
{\bf (Papadakis-Webster binary string (PWBS)):} 
A Papadakis-Webster binary string is the quotient of a Papadakis-Webster integer divided by 99.
\end{definition}

\begin{corollary}
A Papadakis-Webster binary string for integer $D$ is the reverse of the 
digital-borrowing binary indicator sequence for $D-rev(D)$, excluding the leading digit.
\end{corollary}
This can be seen from Eq.\ref{eq-binary} that the first binary value 
in F/99 is $b_0$ for $i=n-1$, the second is $b_1$ for $i=n-2$, etc.,
and the last binary value is $b_{(n-1)}$ for $i=0$.

\begin{corollary}
For an initial integer $D$ of length $n+1$, 
the length of 
the corresponding Papadakis-Webster binary string is $n$.
\end{corollary}
This can be seen 
by the upper limit of summation in Eq.\ref{eq-binary}
($n-1$ instead of $n$). 
Specific examples can
be seen at Table \ref{table1}. Note that the length of the corresponding Papadakis-Webster
integer is either $n+2$ or $n+1$.

Table \ref{table1} shows all Papadakis-Webster integers when the initial integer $D < 10^7$.
Besides the well known Papadakis-Webster integers 99 and 1089 when the initial integers
have 2 or 3 digits, the new Papadakis-Webster integers include 9999, 10890, 10989,
99099, 109890, 109989, etc.

\begin{table}[H]  % \normalsize
\begin{center}
All PWIs when the initial integer is less than 10 millions  \\
\begin{tabular}{c|c|cc|c}
\hline
n+1 & num PWI & PWI & PWBS & not allowed binary strings\\
\hline
2 & 1 & 99 & 1 & \\
\hline
3 & 1 & 1089 & 11 &  10\\
\hline
4 & 3 &  9999 & 101 & 100 \\
 & & 10890& 110 & \\
 &  & 10989 & 111 & \\
\hline
5 & 3&  99099 & 1001 & 1000, 1010, 1011, 1100, 1101 \\
&  & 109890& 1110 & \\
&  & 109989 & 1111 & \\
\hline
6 & 8& 990099 & 10001 & 10000, 10010, 10100, 10110, \\
&  & 991089 & 10011& 10111,  11000, 11001, 11101\\
&  & 999999 & 10101 & \\
&  & 1089990& 11010 & \\
&  & 1090089 & 11011 & \\
&  & 1098900 & 11100& \\
&  & 1099890 & 11110 & \\
&  & 1099989& 11111 & \\
\hline
7 & 8 & 9900099 & 100001  & 100000, 100010, 100100, 100101,  \\
 &  & 9901089  & 100011  &100110, 100111,  101000, 101001,  \\
 &  & 10008999  & 101101  & 101010, 101011, 101100, 101110,    \\
 &  & 10890990  & 110010   & 101111, 110000, 110001, 110100,   \\
 &  & 10891089  & 110011  & 110101, 110110, 110111, 111000,     \\
 &  & 10998900  & 111100  & 111001, 111010, 111011, 111101 \\
 &  & 10999890  & 111110 & \\
 &  & 10999989  & 111111  & \\
\hline
8 & 21& (see appendix) & & \\
9 & 21 & (see appendix) & & \\
10 & 55 & & & \\
11 & 55 & & & \\
12 & 144 & & & \\
13 & 144 & & & \\
$\cdots$ & $\cdots$ & & & \\
\hline
$(n+1)$ even & $F_{n+1}$ & & & \\
$(n+1)$ odd & $F_{n}$ & & & \\
\hline
\end{tabular}
\end{center}
\caption{ \label{table1}
All PWI and the corresponding PWBS when the
starting integer's length (n+1, for $a_n a_{n-1} \cdots a_1 a_0$) 
is 1-7. Binary strings that are not PWBS are also 
listed (last column). The number of PWI's as a function of $n$
follows a (partial) Fibonacci sequence (1,1,3,3,8,8,21,21...).
The PWI and the corresponding PWBS for n+1=8,9
are included in the Appendix A.2.
}
\end{table}

The number of Papadakis-Webster integers increase gradually with the
number of digits of the integer. Webster shows that the number of unique Papadakis-Webster
integers as a function of digit length $n+1$ is a ``stepwise" Fibonacci sequence \citep{webster}, by
which we mean that the numbers are not 
$F_2=$1, $F_3=$2, $F_4=$3, $F_5=$5, $F_6=$8, $F_7=$13, $\cdots$, 
but 1,1,3,3, 8,8,21,21, $\cdots$ (see Table \ref{table1}).

On one hand, there are infinite numbers of Papadakis-Webster integers, one the other hand,
the percentage of Papadakis-Webster integers out of all possible integers decreases exponentially
as a function of the number of digits $n$:
$5^{-0.5} \phi^n/10^n \approx 6.18^{-n}/\sqrt{5}$ (where $\phi=(1+\sqrt{5})/2 \approx 1.618$ is the golden ratio)

\subsection{
Hairpin pairing rule for Papadakis-Webster binary strings}

\indent

For each Papadakis-Webster integer, we also list the corresponding Papadakis-Webster
binary string in Table \ref{table1}.  Not all binary strings are Papadakis-Webster binary
strings. For example, 10, 100, 1000, 1010, 1011, 1100, 1101, etc. In the following
proposition, we establish the existence of binary strings that are not Papadakis-Webster
binary strings. Note that we will still call 0's and 1's in a PWBS
(decimal) digits, not bits which are binary digits. The reason is that PWBSs are
still defined in the decimal system, not in the binary system. 

\begin{proposition}
\label{prop-not-all}
Not all binary strings are Papadakis-Webster binary string.
\end{proposition}
\begin{proof}
We prove it by two counterexamples. Suppose $n$ is an even number, and the
number of digits $n+1$ is odd. It means that there is a digit 
exactly in the center position $a_{n/2}$. We will show that $1 \cdots 10 \cdots$,
where 1 is the digit-borrowing indicator value for $D-rev(D)$  at 
position $n/2-1$, and 0 is that at the position $n/2$, can not be a
Papadakis-Webster binary string. Considering the following digit-borrowing pattern:
\begin{equation}
\begin{matrix}
 & a_n & a_{n-1} & \cdots & a_{n/2} & a_{n/2-1} & \cdots & a_1 & a_0 \\
 -) & a_0 & a_1 & \cdots & a_{n/2} & a_{n/2+1} & \cdots & a_{n-1} & a_n \\
\hline
 b_i  & (0)  &  & \cdots & 0 & 1 & \cdots &  &  1 \\
\end{matrix}
\end{equation}
which means $a_{n/2-1} < a_{n/2+1}$ or $a_{n/2-1}-1 < a_{n/2+1}$,
(need to borrow) and $a_{n/2}-1 \ge a_{n/2}$ (no need to borrow). But the
latter inequality, implying $-1 \ge 0$,  is impossible. Therefore, the
reverse of the digit-borrowing binary string (after removing the leading 
digit 0,
in parenthesis), $1 \cdots 10 \cdots$, cannot be a Papadakis-Webster binary string.  

In the second example, suppose $n$ is odd (then the number of digits is even).
The central two positions are $(n+1)/2$ and $(n-1)/2$. Considering the following
digit-borrowing patterns: 
\begin{equation}
\begin{matrix}
    & a_n & a_{n-1} & \cdots & a_{(n+1)/2} & a_{(n-1)/2} & a_{(n-3)/2} & \cdots & a_1 & a_0 \\
-)  & a_0 & a_1 & \cdots & a_{(n-1)/2} & a_{(n+1)/2} & a_{(n+3)/2} & \cdots & a_{n-1} & a_n \\
\hline
b_i   & (0)  &  & \cdots & 0 & 0 & 1 & \cdots & & 1 \\
\end{matrix}
\end{equation}
This implies $a_{(n-1)/2}-1 \ge a_{(n+1)/2}$ and $ a_{(n+1)/2} \ge a_{(n-1)/2}$.
Adding the two lead to $-1 \ge 0$, which is impossible. Therefore, the reverse of
the digit-borrowing binary sequence can not be a Papadakis-Webster binary string.
$\blacksquare$
\end{proof}

\begin{definition}
{\bf (Paired positions in Papadakis-Webster binary string)} Denote Papadakis-Webster binary string
as $B_{n-1} B_{n-2} \cdots B_{n-k-1} \cdots B_{k-1} \cdots B_1 B_0$, which is the reverse of
the digit-borrowing sequence ($b_n b_{(n-1)} \cdots b_1 b_0$), after removing $b_n$. The positions
$k-1$ and $n-k-1 (k=1, 2, \cdots n-1$) are defined as paired positions.
\end{definition}
The following graph shows the correspondence between $\{ a_i\}$, $\{ b_i \}$, and $\{ B_k \}$:
\normalsize
\begin{equation}
\begin{matrix}
 & a_n & a_{n-1} & \cdots & a_{n-k} & \cdots &  a_k & \cdots & a_1 & a_0 \\
 (-) & a_0 & a_1 & \cdots & a_k & \cdots &  a_{n-k} & \cdots & a_{n-1} & a_n \\
\hline
\mbox{(digit-borrow)} & b_n  & b_{(n-1)}  & \cdots & b_{(n-k)} & \cdots &  b_k & \cdots & b_1 & b_0 \\
\mbox{(reverse PWBS)} & \mbox{(removed)}   & B_0  & \cdots & B_{k-1} & \cdots &  B_{n-k-1} & \cdots & B_{n-2} & B_{n-1} \\
\end{matrix}
\end{equation}
\large
Note that the leading digit $B_{n-1}$ does not have a paired position.

\begin{theorem}
\label{theorem-hairpin}
Papadakis-Webster binary string can not have the same value at paired positions 
except for two situations: the digits preceding them are both 1 and their own values
are both 1, or, the digits preceding them are both 0 and their own values are both 0.
\end{theorem}

\begin{proof}
Since the two paired positions are involved in the subtraction operation of $a_k$ and $a_{n-k}$
when the two are switched, generally speaking they can not both borrow digit, or both not borrow
digit. For example, if they both borrow, and without their neighbors borrowing from them,
then  $a_k < a_{n-k}$ and $a_{n-k} < a_k$, which is impossible. Similarly, if they
both do not borrow, whereas their neighbors borrowing from them, then
$a_k-1 \ge a_{n-k}$ and $a_{n-k}-1 \ge a_k$, which implies $a_k-1 \ge a_{n-k} \ge a_k+1$,
or $-1 \ge +1$, also impossible. All other combinations can be shown similarly.

For our two exceptions, the first is equivalent to $a_k-1 < a_{n-k}$ and $a_{n-k}-1 < a_k$,
which implies $a_k-1 < a_{n-k}  < a_k+1$, with a unique solution $a_k=a_{n-k}$.
In the second exception, $a_k \ge a_{n-k}$ and $a_{n-k} \ge a_k$ also has a unique
solution of $a_k=a_{n-k}$.
$\blacksquare$
\end{proof}

Theorem \ref{theorem-hairpin} provides an algorithm to generate all Papadakis-Webster binary
strings:
\begin{enumerate}
\item
Start from $B_{n-1}=1$.
\item
Pick the next digit at $B_{n-2}$. If $B_{n-2}=1$, continue to $B_{n-3}$.
If $B_{n-2}=0$, set the digit at the pairing position with $n-2$ (which is $B_0$) to 1.
\item
For any k $(k=1, 2, \cdots, n-1)$, when $B_{n-k-1}=0$ but $B_{n-k}=1$, set $B_{k-1}=1$;
when $B_{n-k-1}=1$ but $B_{n-k}=0$, set $B_{k-1}=0$;
when $B_{n-k-1}=B_{n-k}$, move to the next digit. 
\end{enumerate}
Note: (1) this procedure will be carried out even when the index passes the mid-point.
In other words, even if the pairing position is on the left of the current position,
the rule needs to be checked; (2) if the pairing position is the same as the current position,
the rule still needs to be checked.

Theorem \ref{theorem-hairpin} also provides a way to check if a binary
sequence is PWBS or not. Given a binary string started with 1: 
$x_{n-1} x_{n-2} x_{n-3}\cdots x_1 x_0$; removing the leading digit $x_{n-1}=1$, then pairing
the remaining digits around the middle position (pairing $x_{n-2}$ with $x_0$, $x_{n-3}$ with $x_1$, etc. 
The digits in the pairing position should not be the same unless that same value is 1 and
the digits on their left are also 1, or, that same value is 0 and the digits on their left
are also 0.

Fig.\ref{fig1}(A)(B) show two examples. Removing the leading 1 from 10111,
resulting in 0111. The middle point is the space between second 1 and the third 1.
Folding the string around the middle point, with first 0 pairing with the last 1,
and second 1 pairing with the third 1 (Fig.\ref{fig1}(A). It is fine when the
two pairing digits have different values, but when the second 1 is the same as the third 1,
the digits on their left should also be 1s -- they are not, so it is not a PWBS.

In the second example, 110010011100110110,  removing the lead 1, the middle
position is the 1 that separates 7 digits on the left and 7 on the right 
(see Fig.\ref{fig1}(B)).
The middle position pairs with itself, and their (it's) left digit has to be 1 --
indeed it is, so there is no problem. All other pairing digits have different values,
and again there is no problem. Therefore, this sequence is a PWBS. 
Because the folding and pairing of digits in Fig.\ref{fig1} 
is reminiscent of a secondary structures of RNA \citep{holley,grover}, the
hairpin or stem loop, we call Theorem \ref{theorem-hairpin}
``hairpin rule for Papadakis-Webster binary strings".

\begin{figure}[H]
 \begin{center}
  \begin{turn}{-90}
  \end{turn}
 \includegraphics[width=0.9\textwidth]{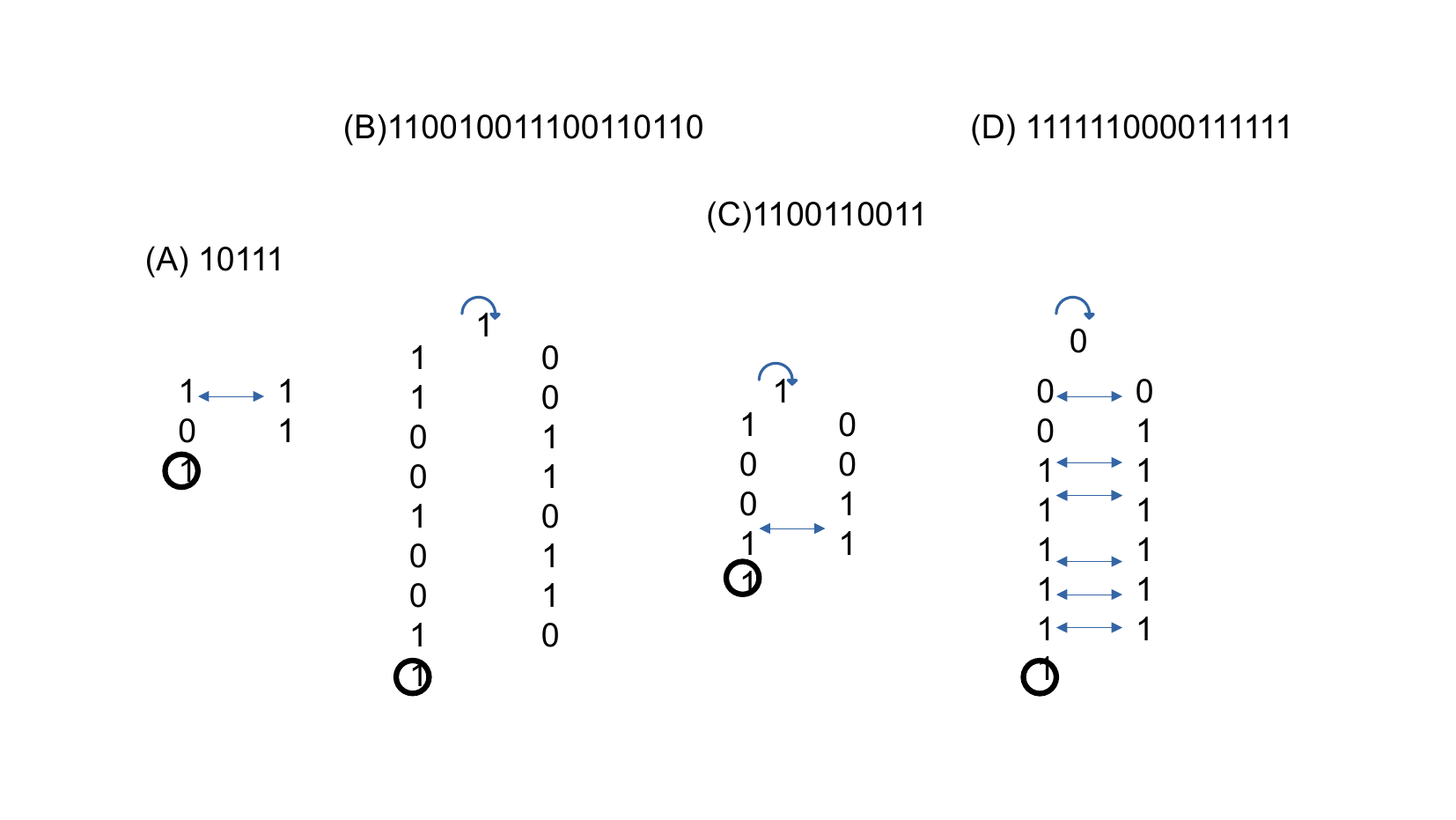}
\end{center}
\caption{ \label{fig1}
Illustration of the hairpin pairing rule for PWBS's. To check if a binary string is
PWBS, remove the leading 1 from the sequence, then fold the rest of the sequence
around its middle position (either between two middle digits, or the middle digit itself).
If all pairing digits have different values, it is a PWBS. If pairing digits are both 1's,
as long as their left digits are also both 1's, it is fine; similarly, if pairing
digits are both 0's, it will be fine if the digits on their left are also both 0's.
For the self-pairing digit in the middle position, its left digit should
be the same as itself in order to be PWBS.
Using this rule, (B)(C)(D) are PWBS, whereas (A) is not PWBS.
The double-head arrows mark the digits in pairing position that have
the same value. The single-head arrows mark the self-pairing digits.
}
\end{figure}

\section{Extension of 1089 trick to any number of iterations }

\subsection{Introducing a new iteration or mapping or dynamical system \label{sec-map}}
\indent

As part of the extension of the ``1089 trick", we introduce the following 
iteration on integers (steps $i=1,2, \cdots \infty$) with two options:
\begin{eqnarray}
\label{eq-map}
D_{i+1} =
\left\{
\begin{matrix}
D_i - rev(D_i) & \mbox{if $rev(D_i) < D_i$} \\
D_i + rev(D_i) & \mbox{if $rev(D_i) \ge D_i$ } \\
\end{matrix}
\right.
\end{eqnarray}
where the emphasis is given to the way we transform a two-step procedure
to an endless dynamical system. The original 1089 trick has a caveat that the number 
of digits after subtraction should not be smaller (otherwise the trick
would not work, see Subsection \ref{sec-caveat}). Our iteration of Eq.\ref{eq-map}
is indifferent to this caveat, and our results are more general.

Note that the subtraction is carried out when the reverse is strictly smaller than the
original integer. If the condition $<$  is changed to $\le$, a palindrome integer
will iterate to zero (and forever be zero). We keep the current conditions of 
Eq.\ref{eq-map} in order to produce more interesting dynamics. 
To simplify the citing, we call Eq.\ref{eq-map} 
``iterative digital reversal" or IDR in the remaining of the paper.

For any dynamical systems, there can be these possible dynamical behaviors:
\begin{enumerate} 
\item
Fixed points: $\exists I \in \mathbb{N}^+$ so that $\forall i > I$, $D_{i+1}=D_i$. 
\item
2-cycles:  $\exists I \in \mathbb{N}^+$ so that $\forall i > I$, $D_{i+2}=D_i$.
\item
Periodic with higher cycle length:
$\exists I \in \mathbb{N}$ so that $\forall i > I$,
 $D_{i+p}=D_i$  ($p \in \mathbb{N}^+ > 2$).
\item
Non-periodic: $\nexists I, p  \in \mathbb{N}^+$,  so that
$D_{i+p}=D_i$ ($i > I$).
\end{enumerate} 
The item-1 in the above list, fixed point, is impossible. For IDR, if $D_{i+1}=D_i$,
then $rev(D_i)$=0, or $D_i=0$. However, we specified the conditions used in 
Eq.\ref{eq-map} so that $D_i$ can never be zero.

Because limiting sets of a dynamical system are called ``attractors",
those of our 1089-trick-inspired map, IDR, can be called ``1089 attractors",
whereas the end-integers from the original 1089 trick (and its extension to
any number of digits) might be considered as 1089 attractors in a narrow sense. 
It explains the words ``1089 attractor" in our title. 
In the next subsection, we will study the digit patterns in the 2-cycle attractors.

\subsection{Constraints on the integers in the 2-cycle attractor}

\indent

\begin{proposition}
\label{prop-tenD}
If $D_I$ ($I \in \mathbb{N}^+$) and $D_{I+1} > D_I$ are two integers in the limiting 
attractor of IDR, then $D_{I+1}= 10D_I$.
\end{proposition}
\begin{proof}
Because $D_{I+1} > D_I$, the second option in Eq.\ref{eq-map} is used to map
$D_I$ to $D_{I+1}$, and because $D_{I+2}=D_I < D_{I+1}$, the first option in 
Eq.\ref{eq-map} is used to map $D_{I+1}$ to $D_{I+2}$:
\begin{eqnarray}
D_{I+2} &=& D_{I+1} - rev(D_{I+1}) \nonumber \\
 &=& D_I +rev(D_I) - rev( D_{I+1})
\end{eqnarray}
Because these are 2-cycle elements, $D_{I+2}=D_I$, therefore
$rev(D_I)=rev(D_{{I+1}})$. There are only two possibilities: the first is that $D_I=D_{I+1}$,
which can not be correct
as there is no fixed-point solution to Eq.\ref{eq-map} and
we only consider 2-cycle here. The second 
possibility is that $D_{I+1}$ is $D_I$ following by a string of 0s.
Eq.\ref{eq-map} can only increase or decrease the length of an integer by 1, therefore,
the number of trailing zeros can only be one. In other words, $D_{I+1}=D_I \times 10$.
\end{proof}

\begin{corollary}
\label{cor-plus-minus}
Among the two integers in the 2-cycle of limiting attractor of IDR,
the subtraction and addition options in Eq.\ref{eq-map} are alternately used.
\end{corollary}
It is because  $D_I \rightarrow D_{I+1}$ increases the integer value
so addition part of Eq.\ref{eq-map} must be used, and 
$D_{I+1} \rightarrow D_{I}$ decreases the integer value, thus the
subtraction is used.

\begin{corollary}
\label{cor-palin}
For two elements in the 2-cycle attractor of IDR, $D_{I+1} > D_I$, then 
$rev(D_I)/D_I=9$. Inversely, if rev(D)= 9D, then
D and 10D are the two 2-cycle elements of IDR. 
\end{corollary}
Since $D_{I+1}= D_I+rev(D_I)= D_I \times 10$ according to Proposition \ref{prop-tenD}, 
$rev(D_I)= 9D_I$. 
Similarly, if rev(D) = 9D, then D+rev(D) = 10D, 10 D-rev(D) = D,
so D and 10 D are the 2-cycle elements of IDR.
The fact that the digital reversal of 1089 is divisible by itself
was mentioned in \citep{hardy}. 

\begin{definition}
{\bf (Palintiple)} If an integer D, whose digital reversal rev(D) is divisable by itself,
i.e., rev(D)/D = k, where $k$ is a positive integer, then D is called a ($k$-)palintiple.
\end{definition}
Corollary \ref{cor-palin} is equivalent to the statements that 
(1) (the smaller member of) any 2-cycle of IDR is 9-palintiple;
(2) any 9-palintiple integer, as well as its 10-multiple, are 2-cycle elements of IDR. 

Note that other publications may define 9801 as palintiple.
In our definition above, 1089 is a palintiple. Our definition is more convenient
in the context of IDR.

\begin{proposition}
\label{prop-1089n}
If $D_I$ ($I \in \mathbb{N}^+$) and $D_{I+1} > D_I$ are two integers in the limiting
attractor of IDR, the first two digits of $D_I$ are 1 and 0, and
the last two digits of $D_I$ are 8 and 9.
\end{proposition}
\begin{proof}
Using the proposition \ref{prop-tenD} and corollary \ref{cor-plus-minus},
$D_I +rev(D_I)= 10 \times D_I$. If $D_I=\sum_{i=0}^n a_i 10^i$, we have
$\sum_{i=0}^n a_i 10^i + \sum_{i=0}^n a_{n-i} 10^i = \sum_{i=0}^n a_i 10^{i+1}$.
To equate the coefficients on both sides, we use the digital
carryover binary sequence $\{ c_i\}$  defined in Eq.\ref{eq-carry}: 
\begin{equation}
\label{eq-tenD}
a_i + a_{n-i} -10c_i +c_{i-1}= a_{i-1}
\end{equation}
We can write Eq.\ref{eq-tenD} explicitly for the leading and trailing digits:
\normalsize
\begin{equation}
 \begin{matrix}
column & n+1 & n  & n-1 &  \cdots & 1 & 0\\
\hline
a_i: &  & a_n=1^{(1)}   & a_{n-1} \in (0,1)^{(3)}=0^{(4)} &   \cdots & a_1 \in (8,7)^{(3)}=8^{(4)}  & a_0=9^{(2)}\\
a_{n-i}: & & a_0=9^{(2)}  &a_1 \in (8,7)^{(3)}=8^{(4)} & \cdots  &
 a_{n-1} \in (0,1)^{(3)}=0^{(4)} & a_n=1^{(1)} \\
-10 c_i: &  & -10c_n=-10^{(1)}  &-10c_{n-1}=0^{(4)} & \cdots & -10c_1=0^{(3)} & -10c_0=-10^{(2)} \\
+c_{i-1}: &c_n=1^{(1)} & c_{n-1}=0^{(4)}&  c_{n-2} & \cdots  & c_0=1^{(2)} & 0 \\
\hline
a_{i-1} & a_n=1^{(1)} & a_{n-1} \in (0,1)^{(3)}=0^{(4)} & a_{n-2} &  \cdots  &  a_0=9^{(2)} & 0  \nonumber \\
 \end{matrix}
\end{equation}
\large
We can derive $(a_n, a_{n-1}, a_1, a_0)=(1,0,8,9)$,
as well as $(c_n, c_{n-1}, c_1,c_0)=(1,0,0,1)$, 
by the following steps (matching the superscripts above:
\begin{itemize}
\item
(1) (column-(n+1))
$c_n$ can not be zero because we know $D_{I+1}= 10D_I$
having one more digit than $D_I$. Therefore $c_n=1$, resulting to $a_n=1$.
\item
(2) (column-0)
$a_0+1-10c_0=0$. $c_0$ can not be zero, because otherwise we have
$a_0+1=0$, with a negative solution for $a_0$.  Therefore $c_0=1$, resulting in $a_0=9$.
\item
(3) (column-1 and column-n)
$a_1+a_{n-1}=8+10c_1$. $c_1$ can not be 1, because it implies $a_1=a_{n-1}=9$.
but from column-n, $a_{n-1}=c_{n-1}$ has only a binary value $\in (0,1)$. Therefore, 
$c_1=0$.
\item
(4) (from column-(n-1) and column-1)
$8 -10c_{n-1}+c_{n-2}=a_{n-2}$. $c_{n-1}$ can not be 1 because it implies
$8-10 + c_{n-2}=a_{n-2}$, or $a_{n-2}$ being negative. Therefore, $c_{n-1}=0$; 
then from column-n, we have $a_{n-1}=0$.
Because $a_1+a_{n-1}=8$, we have $a_1=8$.
$\blacksquare$
\end{itemize}
\end{proof}

\subsection{The 10(9)$_L$89 ($L \ge 0$) motif in 2-cycle integers
and proof that the quotients dividing by 99 are Papadakis-Webster binary strings \label{sec-10989}}

\indent

Table \ref{table2} shows all limiting 2-cycles of IDR when the initial integers
have length 1-7, as well as some examples with even larger initial integers. Not only
these confirm our proposition \ref{prop-1089n} that the limiting 2-cycle integers
start with 10 and end with 89, but there are more specific patterns.
For the 2-cycle integers, there is a fundamental building block of the form
10(9)$_L$89 where the integer $L \ge 0 $: these can be a single such block, or
a symmetric arrangement of multiple blocks. 
This can be summarized by the following proposition:
\begin{proposition}
\label{prop-v1}
If $M_L$ is of a form of 10(9)$_L$89  (i.e., the middle 9 repeats $L \ge 0$ times),
then $M_L$ and other symmetric forms constructed from $M_L$ and padding zeros:
$M_L (0)_K M_L$ (where $K  \ge 0$), 
or $M_{L_1} (0)_{H_1} M_{L_2} \cdots (0)_{K_m} \cdots  M_{L_2} (0)_{K_1} M_{L_1}$ 
(where integers $L_1, L_2, \cdots K_1, K_2, \cdots K_m \ge 0$), 
or $M_{L_1} (0)_{K_1} M_{L_2} \cdots (M)_{L_m} \cdots  M_{L_2} (0)_{K_1} M_{L_1}$ 
(where $L_1, L_2, \cdots L_m, K_1, K_2, \cdots   \ge 0$), 
are limiting 2-cycles of IDR.
\end{proposition}
\begin{proof}
We first prove the case of $M_{L} (0)_{K} M_{L}$ when $K > 0$ and $L > 0$.
Since the first digit is 1 and the  last digit is 9, the addition is carried out first:
\begin{equation}
\label{eq-prove-1089-1}
\begin{matrix}
 & 10(9)_L 89 & (0)_K & 10 (9)_L 89 \\
+) & 98(9)_L 01 & (0)_K & 98 (9)_L 01 \\
\hline
 & 109 (9)_{L-1} 890 &  (0)_{K-1} & 109 (9)_{L-1} 890\\
\end{matrix}
\end{equation}
The next step will be subtraction because the last digit is 0:
\begin{equation}
\label{eq-prove-1089-2}
\begin{matrix}
 & 109 (9)_{L-1} 890 &  (0)_{K-1} & 109 (9)_{L-1} 890\\
-) & 098 (9)_{L-1} 901 &  (0)_{K-1} & 089 (9)_{L-1} 901\\
\hline
 & 010 (9)_{L-1} 989 & (0)_{K-1} & 010 (9)_{L-1} 989\\
= & 10 (9)_{L} 89 & (0)_{K} & 10 (9)_{L} 89
\end{matrix}
\end{equation}

When $K=0$ or/and $L=0$, 
it can be checked that the result remains to be correct,
due to the tailing 0 and/or leading 9 from the neighboring digits of
a repeating unit.
Proofs for the case of  other symmetric combinations of the motifs can be shown 
similarly.$\blacksquare$ 
\end{proof}

Our numerical runs, exhaustive for up to 
9-digit input integers, and then sampling several millions of randomly 
selected input integers of higher digital length, have convincingly 
indicated  that only integers in the form described by Proposition \ref{prop-v1}
are the lower members of the limiting 2-cycle, while the
second (and higher) number is invariably the tenfold multiple of the first.

\begin{proposition}
\label{prop-v1-only}
The integers described in Proposition \ref{prop-v1},
 as well as their 10-multiples,
are the only limiting 2-cycle elements of IDR. 
\end{proposition}
Outline of a proof of Proposition \ref{prop-v1-only}:
similar to the proof of Proposition \ref{prop-1089n} where we show
that  the property of $D_{I+1}=10D_I$ forces the first two digits to be 0,1, and the
last digits to be 8,9, we can continue to examine the constrain towards the
middle of the sequence. For example, one can show that the first three and the
last three digits can either (1,0,9,$\cdots$,9,8,9), or
(1,0,8,$\cdots$,0,8,9). The first 4 and last 4 digits in the first situation
would be (1,0,9,9, $\cdots$, 9,9,8,9), and in the second situation
(1,0,8,9, $\cdots$, 1,0,8,9), etc.
Once the uniqueness of $M_L=10 (9)_L 89$ as the fundamental 2-cycle motif
is established, its symmetric concatenation with padding zeros in various forms
can be shown also to be 2-cycles, similar to the proof for Proposition \ref{prop-v1}.

Alternatively, according to the Corollary \ref{cor-palin}, we only need to prove
that the patterns described in Proposition \ref{prop-v1} are the only integers
for 9-palintiples. This proof can be found in \citep{hoey} (by Dan Hoey)
and in \citep{webster12} (by Roger Webster and Gareth Williams).

Table \ref{table2} also shows that when $D_I$ divided by 99, the quotient is a binary
string. This binary string has some particular properties:
it is symmetric with respect to the center, and the length of 1-blocks
or 0-blocks is at least two. We have the following proposition:
\begin{proposition}
\label{prop-v2}
These binary strings after multiplied by 99 are limiting 2-cycle
elements of IDR: $(1)_{L} (L \ge 2)$, 
or $(1)_L (0)_K (1)_L (L,K \ge 2)$,
or $(1)_{L_1} (0)_{K_1} (1)_{L_2}\cdots(0)_{K_m} \cdots (1)_{L_2}(0)_{K_1}(1)_{L_1} $ \\
$(L_1, L_2, \cdots, K_1, K_2 \cdots K_m \ge 2)$,
or $(1)_{L_1} (0)_{K_1} (1)_{L_2}\cdots (1)_{L_m} \cdots(1)_{L_2}(0)_{K_1}(1)_{L_1} $ \\
$ (L_1, L_2, \cdots L_m, K_1, K_2, \cdots \ge 2)$,
and these are the only form of 99-quotient of 2-cycle elements of IDR.
\end{proposition}
\begin{proof}
It is not difficult to show that 10(9)$_{L-2}$89 = 99 $\times (1)_{L}$,
1089$(0)_{L-2}$1089 = 99 $\times 11(0)_{L}11$, and combining the two,
10(9)$_{L-2}$89(0)$_{K-2}$10(9)$_L$89= 99 $\times (1)_{L} (0)_{K} (1)_{L+2}$.
Since the number 9s between 10 and 89, $L-2 \ge0$, we have $L \ge 0$.
Since we require the number of 0s between motifs, $K-2 \ge 2$, then $K \ge 2$.
Other more complicated situations can be proven similarly.
$\blacksquare$
\end{proof}
Note that the minimum 0-block length or 1-block length is 2, versus 
no minimum length requirement of zero-padding 
in Proposition \ref{prop-v1}.

Comparing Table \ref{table1} and \ref{table2}, it can be seen that not all
PWI's, in fact very few of them (those in Table \ref{table1}),
 can be limiting 2-cycle integers
(those in Table \ref{table2}), due to the special requirements
for 2-cycle integers (e.g. symmetric binary string). On the other hand,
all 2-cycle integers (in Table \ref{table2}) are PWI's (in Table \ref{table1}), 
even though these do not satisfied the condition (i.e., $D-rev(D)$ having the
same number of digits as $D$) in the proof of PWI (proposition \ref{prop-borrow}).
For example, 10890-09801=1089 loses one digit, and according to section 2.3,
1089+9801 is not guaranteed to be PWI.
We propose the following theorem: 
\begin{theorem}
\label{theorem-sym-blocks}
The binary strings described in Proposition  \ref{prop-v2}, i.e., symmetric
arrangement of 0-block and 1-block, whose lengths are larger or equal to 2,
are PWBS.
\end{theorem}

\begin{proof}
We present a proof by examples. In any binary sequence of this symmetric type, because of 
the symmetry and independently of the specific arrangement of any considered 
case, removing the first leading digit 1, then folding around the middle point, 
one can visually see that whenever the pairing digits have the same value, 
their left digits also have the same value, as we may see in
Fig.\ref{fig1}(C) and (D)
where two examples of such binary sequences:
1100110011, and 1111110000111111 are shown. 
$\blacksquare$
\end{proof}
Theorem \ref{theorem-sym-blocks} shows that
all limiting 2-cycle integers of IDR are PWI.

\textcolor{blue}{
}

\begin{table}[H]  % \normalsize
\begin{center}
All 2-cycle elements of IDR when the initial integer is less than
10 millions  \\
\begin{tabular}{c|c|cc|cc}
\hline
init seq & &   & $(1)_k$ or \\
length & 2-cycle attractor & divided by 99 & $(1)_{k1}(0)_{k}(1)_{k1}$ or\\
 & & & $(1)_{k1}(0)_{k2}(1)_k(0)_{k2}(1)_{k1}$ \\
\hline
1 and 2 & 1089, 10890 & 11, 110 & k=2\\
3 & + 10989, 109890 & 111, 1110 & k=3\\ 
 &109989, 1099890 & 1111, 11110 & k=4\\ 
4 & + 1099989,10999890 & 11111, 111110 & k=5 \\
5 & + 10999989, 109999890 &  111111, 1111110 & k=6\\
 & 10891089, 108910890 & 110011, 1100110 & $k_1=2,k=2$\\
 & 108901089,1089010890 & 1100011, 11000110 & $k_1=2,k=3$\\
 & 1098910989, 10989109890 & 11100111, 111001110 & $k_1=3,k=2$\\
6  & +109999989, 1099999890   & 1111111, 11111110 & k=7\\
 &  1089001089, 10890010890 & 11000011, 110000110 & $k_1=2,k=4$\\
 &  1099999989, 10999999890 &11111111, 111111110 & k=8\\
7 & +10890001089,108900010890 & 110000011, 1100000110 & $k_1=2,k=5$\\
 & 10989010989,109890109890 &111000111, 1110001110  & $k_1=3,k=3$\\
 & 10999999989, 109999999890 & 111111111, 1111111110 & k=9 \\
 & 108900001089,1089000010890 & 1100000011, 11000000110 & $k_1=2,k=6$\\
 & 109999999989,1099999999890 & 1111111111, 11111111110 & k=10 \\
 & 109890010989,1098900109890 & 1110000111, 11100001110 & $k_1=3,k=4$ \\
 & 109989109989,1099891099890 &1111001111, 11110011110 & $k_1=4,k=2$ \\
 & 1099890109989,10998901099890 & 11110001111, 111100011110 &$k_1=4,k=3$ \\
 & 1089000001089,10890000010890 &11000000011,110000000110 & $k_1=2,k=7$ \\
\hline
 & 108910891089,1089108910890 & 1100110011, 11001100110 & $k_1=2,k_2=2,k=2$\\
\hline
 & 108901098901089,1089010989010890 & 1100011100011,11000111000110 & $k_1=2,k_2=3,k=3$ \\
\hline
\end{tabular}
\end{center}
\caption{ \label{table2}
Integers in the limiting 2-cycle of IDR when the length
of the initial integers is 1-7. The ``+" sign means ``plus all 2-cycle
attractor integers already obtained  when the initial integer length is
lower". The third column shows the quotients when these 2-cycle
integers are divided by 99. The last column is an attempt to summarize
the pattern of the binary strings in the third column. The last two
rows show examples for constructing 2-cycle integers by
juxtaposition of repeated copies of previously known 2-cycle integers.
}
\end{table}

\newpage

\subsection{ Construction of an infinite number of $p$-cycles ($p > 2$) }

\indent

Two-cycles are not the only limiting attractors of IDR.
Other limiting cycles are relatively rare but exist. 
The simplest $p$-cycle
for $p>2$ we have observed is a 12-cycle listed in Table \ref{table3}(A). 
This 12-cycle can be found by starting
the iteration from a very small number, $D_0=$158. The integers in this 12-cycle
are listed in Table \ref{table3}(A). 

During the iteration of Eq.\ref{eq-map}, 
integers become 99-divisible, usually well before becoming cyclic,
after the sequence of integers passes through
one subtraction and one addition which led to a PWI.
Once the integer becomes 99-divisible, its subsequent integers by iteration
are also 99-divisible. We can argue in the following: if an integer is 9-divisible,
the sum of its digits is divisible by 9. This feature is preserved by digital
reversal, subtraction, and addition. Therefore, once an integer is 9-divisible,
it will continue to be 9-divisible after applying Eq.\ref{eq-map}. Similarly,
if an integer is 11-divisible, the sum of its digits with alternating signs,
$\sum_{i=0}^n (-1)^i a_i$ is divisible by 11. This feature also will not be
affected by digital reversal, subtraction, and addition. Combining the two,
the 99-divisibility is preserved by Eq.\ref{eq-map}.

However, 99-divisibility do not necessarily imply that these integers are PWI.
Among the 12-cycle elements in Table \ref{table3}(A), only 4 out of 12 
quotient after dividing 99 (``99-quotient") are binary sequences,
and only 3 of them are PWBS, confirmed by theorem \ref{theorem-hairpin}. 
There are no more common factors besides 99, though eleven out of 12 
integers are divisible by 11$\times$99. 

The next cycle length we are examining is $p=10$, 
shown in Table \ref{table4}.  Although it has a shorter
cycle length than 12, it was found later in our numerical runs as it requires
a much larger initial integer, which would not have been found if the initial
$D_0$ is small.
Of the 12 99-quotients
of 10-cycle elements, two are binary strings, and both are PWBS.

Moreover, the 10-cycle is not one attractor but a whole class
of them described by \\
109008910(9)$_L$890991089 ($L=0,1,2, \cdots$).
It is similar to the class of limiting 2-cycles described by 10(9)$_L$89.
We have verified for a lot of members of the family their limiting behavior, 
and the generality of the above statement may be shown the same way as
that in subsection 3.3.

\begin{table}[H]  % \normalsize
\begin{center}
(A) An example of a 12-cycle attractor \\
\begin{tabular}{c|r|r}
\hline
i & 12-cycle integer & divided by 99  \\
\hline
1 & 99099 & (PW) 1001= 7*11*13 \\
2 & 198198 & 2002 = 2*7*11*13  \\ 
3 & 1090089 & (PW) 11011 =7*11*11*13 \\ 
4 & 10890990 & (PW) 110010 =3*10*19*193  \\ 
5 & 981189 & 9911 =11*17*53\\ 
6 & 1962378 & 19822 =2*11*17*53 \\ 
7 & 10695069 & 108031 =7*11*23*61 \\ 
8 & 106754670 & 1078330 =10*11*9803 \\ 
9 & 30297069 & 306031 =11*43*647 \\ 
10 & 126376272 & 1276528 =11*16*7253 \\ 
11 & 399049893 & 4030807 =11*366437 \\
12 & 108900 & (not PW) 1100 =10*10*11 \\ 
\hline
13=1 & 99099 & 1001 =7*11*13\\
\hline
\end{tabular}

\vspace{0.4in}

(B) Another 12-cycle constructed by padding two 99099 separated by four zeros\\
\begin{tabular}{c|r|r}
\hline
i & 12-cycle integer & divided by 99   \\
\hline
1 & 99099000099099 & (PW) 1001000001001   \\
2 & 198198000198198 & 2002000002002     \\
3 & 1090089001090089 & (PW) 11011000011011    \\ 
4 & 10890990010890990 & (PW) 110010000110010   \\
5 & 981189000981189 &9911000009911    \\
6 & 1962378001962378 & 19822000019822   \\
7 & 10695069010695069 & 108031000108031    \\
8 & 106754670106754670 & 1078330001078330   \\
9 & 30297069030297069 & 306031000306031   \\
10 & 126376272126376272 & 1276528001276528    \\
11 & 399049893399049893 & 4030807004030807   \\
12 & 108900000108900 & (not PW) 1100000001100   \\
\hline
14=1 & 99099000099099 & (PW) 1001000001001   \\
\hline
\end{tabular}
\end{center}
\caption{ \label{table3}
(A) Integers in the simplest limiting 12-cycle of IDR. The third column
is the quotient by dividing these integers by 99. Of the four binary string
quotients, 3 are PWBS (see Table \ref{table1}) and 1 is not.
Further prime factorization of these 12 integers shows that besides 99, there
are no other common factors. 
(B) Construction of another 12-cycle by concatenating the simplest
12-cycle elements in (A). Each line in (B) corresponds to a line in (A) in that
it has two copies of the simpler element plus padding zeros in-between. 
}
\end{table}

\newpage

Concatenations of $p$-cycle elements by padding certain
number of zeros in-between, are also $p$-cycle elements. For example,
Table \ref{table3}(B)
shows that combining two 12-cycle elements 99099s (as seen in Table \ref{table3}(A))
 with four zeros between them,
99099 0000 99099, is a 12-cycle element. Comparing integers in Table \ref{table3}(A)
and (B), each row can find its correspondence in another table,
though the number of zero-padding might be different (e.g.,
399049893 in Table \ref{table3}(A) vs. 399049893 399049893 in Table \ref{table3}(B)).

The 10-cycle motif,
$M_L$= 1090089  10(9)$_L$89 0991089 ($L=0,1,2,\cdots$), 
can also be concatenated with  padding zeros in-between,
M$_L$ (0)$_K$ M$_L$ ($K=0,1,2, \cdots$), 
which will also be 10-cycles (result not shown). Similar to 2-cycles,
this concatenation can be generalized to other forms, as long as the overall
arrangement of $M_L$ is symmetric and there are enough number of zero spacers.
We propose the following proposition: 
\begin{proposition}
\label{prop-cat}
If $M$ is a $p$-cycle integer, we can use $M$ as a motif to construct
other $p$-cycles, such as $M'=M (0)_K M$,
$M'=M (0)_{K_1} M (0)_{K_2} M  \cdots  M \cdots M (0)_{K_2} M (0)_{K_1}M$,
$M'=M (0)_{K_1} M (0)_{K_2} M  \cdots  (0)_{K_m} \cdots M (0)_{K_2} M (0)_{K_1}M$, etc.
\end{proposition}
Outline of a proof of Proposition \ref{prop-cat}: similar to the proof
of Proposition \ref{prop-v1}, since the concatenation is symmetrically arranged,
each of the motif $M$ will follow its own $p$-cycle dynamics, and the padding
zeros play the role of separating them. Therefore, $M'$ is also a $p$-cycle element.

\begin{table}[H]  % \normalsize
\begin{center}
An example of 10-cycle attractor \\
\begin{tabular}{c|r|r}
\hline
i & 10-cycle integer & divided by 99   \\
\hline
1 & 1090089109890991089 & (PW) 11011001110010011   \\
2 & 10892080098910791990 & 110021011100109010   \\
3 & 972378109902762189 & 9822001110128911  \\
4 & 1953645319804635468 & 19733791109137732   \\
5 & 10599009408940099059 & 107060701100405041  \\
6 & 105698014389430198560 & 1067656711004345440  \\
7 & 39806979406019302059 & 402090701070902041   \\
8 & 134827370466517262952 & 1361892630974921848  \\
9 & 394090086130590991383 & 3980707940713040317  \\
10 & 10890991098910900890 & (PW) 110010011100110110   *66449\\
\hline
11=1 & 1090089109890991089 & 11011001110010011   \\
\hline
\end{tabular}
\end{center}
\caption{ \label{table4}
Integers in one of the limiting 10-cycle of IDR. The third column
is the quotient by dividing these integers by 99. Both binary quotients,
are checked by Theorem \ref{theorem-hairpin} to be PWBS.
}
\end{table}

In all examples from Tables \ref{table3}-\ref{table4}, as well as our
numerical runs of Eq.\ref{eq-map}, there are always at least one
Papadakis-Webster integer in the limiting cycle. We propose the following conjectures: 
\begin{conjecture}
\label{conj-1PWI}
One of the integers in any $p$-cycle of IDR is a Papadakis-Webster integer.
\end{conjecture}
Examples include: 1089 ($p=2$), 1090089109890991089 ($p=10$), and 99099 ($p=12$).  In 
Tables \ref{table3}-\ref{table4}, we have also seen that not all $p$-cycle ($p >2$) integers
are Papadakis-Webster integers, unlike the situation for $p=2$ in Table \ref{table2}.
Consider the repeated applications of subtraction and addition
in Eq.\ref{eq-map}, we would expect that both reaching the cycle and
within the cycle, one subtraction would be followed by an addition in the next step.
Therefore, one would expect a PWI to emerge, and Conjecture \ref{conj-1PWI} may seem to
be natural. However, due to the caveat mentioned in Subsection \ref{sec-caveat},
it is not guaranteed that the condition in Definition \ref{def-PWI} is satisfied. 
While Eq.\ref{eq-map} always preserve 99-divisibility, it may not preserve PWI
membership.

Table \ref{tableA2} shows integers in a 71-cycle. We can use
these integers to check the above hypothesis. Indeed, out of 71 integers
in the limiting set, 12 of them after dividing 99 lead to binary strings.
Of these 12 binary strings, 8 are checked to PWBS by the
hairpin rule (Theorem \ref{theorem-hairpin}). We can also show that
(e.g.) 8820000289999602 (row 68 in Table \ref{tableA2}) can be used
to form another integer 8820000289999602 00000 8820000289999602, with
five padding zeros, is also part of a 71-cycle integer. Same conclusion
can be reached for three copies of the motif 8820000289999602:
8820000289999602 00000 8820000289999602 00000 8820000289999602 (result not shown). 

Although we have not observed other cycle length besides
2, 10, 12, and 71, we hypothesize that there are other cycle lengths:
\begin{conjecture}
Eq.\ref{eq-map} has limiting cycles with cycle length longer than 71.
\end{conjecture}

\subsection{Non-periodic plus diverging trajectories with infinite number of steps}

\indent

We first made a promise in the Introduction 
that when the 1089 trick is extended to any number of digits, or, 
to any number of steps by a subtraction-addition mixture of iterations, 
something similar to the 1089-magic would appear. Indeed, extending to any number of
digits would lead to a very ``privileged" set of Papadakis-Webster integers,
whose membership is limited. Similarly, when extending the number of
steps from 2 to infinity in IDR, the iteration often ends up
to a 2-cycle, whose members are more restricted, to a subset of 
Papadakis-Webster integers. Eq.\ref{eq-map} may also end up to
a $p$-cycle ($p>2$): we have observed $p=10,12, 71$ through extensive
numerical runs. The integers in a $p$-cycle are still restricted,
but less so that those in a 2-cycle. We have not fully characterized the
feature for integers in the set of all $p$-cycles. 

However, a new type of dynamical behavior appears
from IDR. This is the item no.4 listed in Section \ref{sec-map}:
the non-periodic (acyclic) dynamics. However, unlike the chaotic
dynamics in continuous nonlinear systems, the trajectories we have
observed are not wandering irrationally in the $\mathbb{N^+}$ space;
instead, they march to infinitely large integers in a regular fashion. 

Our  numerical run shows that the following integer would
lead to a diverging/non-periodic trajectory: 10(9)$_n$89(0)$_{n}$ with $n \ge 2$.
This set of integers look deceptively similar to 10(9)$_n$89 (with $n \ge 0$)
for the limiting 2-cycle integers. However there are two major differences:
the extra trailing zeros, and the longer padding of 0's between 10 and 89.
The trailing zeros, in particular, destroy the symmetry between 10 and 89,
and the dynamics is no longer 2-cyclic.

It can be shown that 
10(9)$_i$89(0)$_i$ will map to ($\rightarrow$) 
1(0)$_{i+1}$ 8 (9)$_{i+1}$ $\rightarrow$ 
10(9)$_{i}$ 89 (0)$_{i+1}$ $\rightarrow$ 
108(9)$_{i+1}$ 0 (9)$_{i}$ $\rightarrow$ 
10 (9)$_{i+2}$ 89 (0)$_i$ $\rightarrow$ 
1(0)$_{i+1}$ 998 (9)$_{i+1}$  	$\rightarrow$ 
10(9)$_{i+2}$ 89(0)$_{i+1}$  $\rightarrow$ 
109 (0)$_i$ 9890 (9)$_i$, and finally, maps to 
10(9)$_{i+2}$89(0)$_{i+2}$ (9). 
The above representation of these integers
do not highlight the symmetry hidden in the sequence. For that, we propose the following conjecture.

\begin{conjecture}
\label{con-8}
There are an infinite number of diverging/non-periodic trajectoryies of IDR
where the middle digit(s) follow a 8-cycle rhythm: 98, 08, 8, 9, 99, 99, 9, and 9. 
\end{conjecture}
Note that a cyclic pattern in the middle digits do not necessarily imply that
the integers themselves are cyclic. In fact, the flanking digits (both left and right)
change and increase in length during the iteration for diverging trajectories.
To illustrate Conjecture \ref{con-8}, we rewrite the first 16 integers starting
from 10998900: 109\underline{98}900, 100\underline{08}999, 1099\underline{8}9000, 1089\underline{9}9099,
1099\underline{99}8900, 1000\underline{99}8999, 10999\underline{9}89000, 10900\underline{9}89099, and 
10999\underline{98}90000, 10990\underline{08}90099, 109999\underline{8}900000,
109989\underline{9}900099, 109999\underline{99}890000, 109900\underline{99}890099, 1099999\underline{9}8900000 , 1099900\underline{9}8900099.

\section{Discussion}

\indent

While the extension of 1089 trick to any number of digits
is straightforward, the extension to any number of steps need more discussion. Our
Eq.\ref{eq-map} aims at applying both subtraction and addition of an integer
$D$ and its digital reversal $rev(D)$. However, Eq.\ref{eq-map} is not unique.
For example, suppose we change the condition for applying subtraction
from $rev(D_i) < D_i$ to $rev(D_i) \le D_i$:
\begin{eqnarray}
\label{eq-map2}
D_{i+1} =
\left\{
\begin{matrix}
D_i - rev(D_i) & \mbox{if $rev(D_i) \le D_i$} \\
D_i + rev(D_i) & \mbox{if $rev(D_i) > D_i$ } \\
\end{matrix}
\right.
\end{eqnarray}
a palindromic $D_i$ will map to $D_{i+1}=0$ which would be a fixed point. 
Eq.\ref{eq-map2} will not
generate anything that Eq.\ref{eq-map} can not generate, whereas
it would have more 0-fixed-points. Therefore, it is clear that
Eq.\ref{eq-map} has more complex behaviors
than Eq.\ref{eq-map2}.

Eq.\ref{eq-map} provides a paradigm for generating multiple types of dynamics
using a simple rule. The only factor that determine the limiting dynamics
is the initial integer. An integer may have a particular digit borrowing
(for subtracting its reverse, see Eq.\ref{eq-borrow}) and digit carryover 
(for adding its reverse, see Eq.\ref{eq-carry}) pattern, and 
how these two binary sequences determine the limiting 
dynamical behavior is far from obvious.  For a traditional  dynamical system,
the simplest situation is the fixed points, where the eigenvector/eigenvalues 
of the (transpose of) transition matrix could provide a complete solution 
(e.g., \citep{wli-meta}).
In a similar way,
the simplest situation for Eq.\ref{eq-map} is 2-cycles, which we also
know a great deal, aided by Proposition \ref{prop-tenD}.

Adding or subtracting an integer's digital reverse
from the integer itself, especially if done in multiple or repetitive ways,
is an operation leading to integers obeying to very specific
structural restrictions. These include the digital composition,
the appearance of specific digital motifs, as well as 
the divisibility properties of the outcome.
Therefore, the resulting integers, either after only two steps or after 
a very large
number of steps, could fall into a very small set in the $\mathbb{N^+}$ space. 
Indeed, the proportion of integers that are PWI decreases exponentially
with the number of digits. The limiting 2-cycle elements of IDR are
a subset of PWI, whose number is even fewer. Even though there is an 
infinite number of elements in a diverging trajectory, percentage-wise
these still occupy a very small set since there are specific signature
pattern in these integer sequences. 

PWIs and cycle elements of IDR should not be confused
with each other, even though the two are related. If we take all PWIs
from Table \ref{table1} and \ref{tableA1} as initial conditions for IDR,
there are several possible outcomes. If the PWI is
of the form 10(9)$_L$89 or 10(9)$_L$890, after two iterations, they are
mapped to themselves, i.e., these are limiting 2-cycle elements of IDR.
Otherwise, they can be mapped to other PWIs to form a 2-cycle, or
other 12-cycles, $p$-cycles, or diverging trajectories. In other words,
PWIs, being the end result after two steps in non-caveat situations, 
are not necessarily the end result after being iterated infinite number of steps.

Corollary \ref{cor-palin} states that if $D_I$ is the
smaller element in a limiting 2-cycle set of IDR, then $rev(D_I)/D_I=9$,
i.e., a 9-palintiple. 
This connection between palintiples 
\citep{hardy,sutcliffe,beech,pudwell,holt14,kendrick,holt16} and 2-cycle of IDR provides 
another path in discovering sequence pattern of the integers in the limiting set. 
In fact, the motif 10(9)$_L$89 and its extensions  revealed in \citep{hoey,webster12}
are exactly what we observed in numerical experiment of IDR.

Although all the results in this article concern 
decimal integers (i.e., with base 10), similar results can be 
obtained for non-decimal bases. For example, as already pointed out
by Webster \citep{webster}, the divisibility by 99 for decimal PWI
will become divisibility (b+1)(b-1) (where b is the base) for non-decimal PWIs.
For example, if b=2, non-decimal PWIs are always multiples of 3, if b=8,
they are multiples of 63, etc. New cycle lengths have also been
observed for  non-decimal IDR (results not shown). More
results for 1089 trick and IDR on base-b numerical system are included
in Appendix A.4.

Many mathematical programs or calculators can not
carry out arithmetics correctly for very larger integers. We
have observed, for example, that the R software environment
(https://www.r-project.org/) could produce incorrect output
for simple division of a very large number. To aid such
large number arithmetics, we provide a FORTRAN source code
IDR1f.f90, detailed guidelines, and its executable in Windows environment at: 
{\sl https://shorturl.at/PkwJ6 }, 
{\sl https://shorturl.at/0Qje5 }, 
{\sl https://shorturl.at/G1g1y }.

Finally, since a large number of results are presented here, we 
summarize the major ones in Table \ref{table5}.

\begin{table}[H]  % \normalsize
\begin{center}
A summary of the results presented in this article \\
\begin{tabular}{c|c|c}
\hline
situation & name & description \\
\hline
two-step & PWI & multiples of 99 (=PWBS) (theorem 2.2) \\
& PWBS & reverse of the digit borrowing sequence (Eq.5)  \\
& PWBS & hairpin pairing rule (theorem 2.7, Fig.1)  \\
\hline
two-step with caveat & & no conclusion (appendix A.1) \\
\hline
any steps (IDR) & 2-cycle & basic motif 10(9)$_L$89 ($L\ge 0$) (prop 3.4, 3.5)\\
& 2-cycle & symmetric arrangements of basic motif (prop. 3.7)\\
& 2-cycle /99 & symmetric arrangements of 1- and 0-blocks (length $\ge$2)(prop 3.8) \\
& p-cycle  & p=10,12, 71 \\
& 10-cycle & a motif element: 1090089 10(9)$_L$89 0991089 ($L \ge 1$) (sec 3.4) \\
& p-cycle  & symmetric arrangements of p-cycle motifs with padding zeros (prop 3.9) \\
& diverging & 8-cycle rhythm in the middle: 98,08,8,9,99,99,9,9  (conj 3.12)\\
\hline
\end{tabular}
\end{center}
\caption{ \label{table5}
A summary of results presented in this article.
}
\end{table}

{\bf Acknowledgement:}
We would like to thank Dr. Astero Provata for helpful discussions
and reading the manuscript. YA thanks Dr. Marianna Vasilakaki for computer
assistance, and Ms. Maria Metaxa for sharing a copy of C. Papadakis' monograph. 
WL thanks Oliver Clay for helpful comments on the draft.
{\sl
YA would like to dedicate
this work to the memory of Michael Rossos -- a friend, an old gentleman,
a polymath with a special interest in linguistics and ancient religions.
It was him, after reading a mathematics popularization article about Papadakis' work,
who started to experiment with numbers adding or subtracting them after
reversion, conditioning on their relative values. He was already amazed by
how often number 1089 appeared in the output. Through his playful
creative activity, he was the true inventor of IDR almost 30 years ago. }

\section*{Appendices}

\appendix

\renewcommand{\thesection}{A.\arabic{section}}

\setcounter{equation}{0}
\renewcommand{\theequation}{\thesection.\arabic{equation}}

\setcounter{table}{0}
\renewcommand{\thetable}{A.\arabic{table}}

\large

\section{A caveat in applying the 1089 trick} \label{sec-caveat}

\indent

While the condition $a_n > a_0+1$ in 1089 trick is 
sufficient  to ensure that
the integer after the subtraction step, $E$, has the same number
of digits as the starting integer $D$, the condition $a_n=a_0+1$
is not sufficient.
Let's illustrate this point by the following example:
if $D=4193$,  where $a_n$ is equal to  $a_0+1$ 
E=4193-3914=279 has one less digit than $D$. If we treat
the leading 0 as a space-holding digit, $F=0279+9720=9999$ is indeed a Papadakis-Webster
integer. However, most people would consider rev(279)=972, then  $F=279+972=1351$ is
no longer a Papadakis-Webster integer.

It can be seen from Eq.\ref{eq-E}  that $E$ may lose
the leading digit when $a_n-a_0 -b_{(n-1)}$=0 
(if we exclude the situation of $a_n=a_0$, then $a_n = a_0+1$).
Then the upper limit of the summation in Eq.\ref{eq-F} 
is changed from $n$ to $n-1$, and we have:
\begin{eqnarray}
F &=& E+rev(E) \nonumber \\
 &=& \sum_{i=0}^{n-1} (a_i - a_{n-i} +10 b_i - b_{(i-1)}) 10^i
+ \sum_{i=0}^{n-1} (a_{n-i-1} - a_{i+1} +10 b_{(n-i-1)} - b_{(n-i-2)}) 10^i \nonumber 
\end{eqnarray}
Since no $a_i$ terms are canceled, $F$ keeps more information about the original
sequence $\{ a_i \}$. The property of  uniqueness is lost, and there is no 1089 trick.

\section{Papadakis-Webster integers from the initial integers of length 8 and 9}

\begin{table}[H] \footnotesize
\begin{center}
All PWIs when the initial integer is between 10 millions and 1 billion  \\
\begin{tabular}{c|rc|c}
\hline
n+1 &  PWI & PWBS & not allowed binary strings\\
\hline
8 &99000099 &1000001   & 1000000, 1000010, 1000100, 1000101,  \\
 &  99001089  &1000011 & 1000110, 1001000, 1001010, 1001100,  \\
 &  99010989  &1000111 & 1001101, 1001110, 1001111, 1010000,  \\
 &  99099099  &1001001 & 1010001, 1010010, 1010011, 1010100,  \\
 &  99100089  &1001011 & 1010110, 1010111, 1011000, 1011010,  \\
 &  99999999  &1010101 & 1011011, 1011100, 1011110, 1011111,  \\
 &  100089099 &1011001 & 1100000, 1100001, 1100100, 1100101,  \\
 &  100098999 &1011101 & 1100111, 1101001, 1101100, 1101101,  \\
 &  108900990 &1100010 & 1101110, 1101111, 1110000, 1110001,  \\
 &  108901089 &1100011 & 1110010, 1110011, 1110101, 1111001,  \\
 &  108910890 &1100110 & 1111010, 1111011, 1111101\\
 &  108910989 &1100111 & \\
 &  108999990 &1101010 & \\
 &  109000089 &1101011 & \\
 &  109899900 &1110100 & \\
 &  109900890 &1110110 &\\
 &  109900989 &1110111 &\\
 &  109989000 &1111000 & \\
 &  109998900 &1111100 &\\
 &  109999890 &1111110 & \\
 &  109999989 &1111111 & \\
\hline
9 & 990000099 & 10000001 & 
10000000,10000010,10000100,10000101,10000110, \\
 & 990001089  & 10000011 &  10001000,10001001,10001010,10001011,10001100, \\
 & 990010989  & 10000111 &  10001101,10001110,10001111,10010000,10010001, \\
 & 991089099  & 10011001 &  10010010,10010011,10010100,10010101,10010110, \\
 & 991090089  & 10011011 &  10010111,10011000,10011010,10011100,10011101, \\
 & 999909999  & 10100101 &  10011110,10011111,10100000,10100001,10100010, \\
 & 1000989099 & 10111001 &  10100011,10100100,10100110,10100111,10101000, \\
 & 1000998999 & 10111101 &  10101001,10101010,10101011,10101100,10101101, \\
 & 1089000990 & 11000010 &  10101110,10101111,10110000,10110001,10110010, \\
 & 1089001089 & 11000011 & 10110011,10110100,10110101,10110110,10110111, \\
 & 1089010890 & 11000110 &  10111000,10111010,10111011,10111100,10111110, \\
 & 1089010989 & 11000111 &  10111111,11000000,11000001,11000100,11000101, \\
 & 1090089990 & 11011010 &  11001000,11001001,11001010,11001011,11001100, \\
 & 1090090089 & 11011011 &  11001101,11001110,11001111,11010000,11010001, \\
 & 1098909900 & 11100100 &  11010010,11010011,11010100,11010101,11010110, \\
 & 1098910890 & 11100110 &  11010111,11011000,11011001,11011100,11011101, \\
 & 1098910989 & 11100111 &  11011110,11011111,11100000,11100001,11100010, \\
 & 1099989000 & 11111000 &  11100011,11100101,11101000,11101001,11101010, \\
 & 1099998900 & 11111100 &  11101011,11101100,11101101,11101110,11101111, \\
 & 1099999890 & 11111110 &  11110000,11110001,11110010,11110011,11110100, 11110101, \\
 & 1099999989 & 11111111 &  11110110,11110111,11111001,11111010,11111011,11111101 \\
\hline
\end{tabular}
\end{center}
\caption{\label{tableA1}
Extension of Table \ref{table1}: Papadakis-Webster integers (PWI)
when the initial integers that start the two-step operation have 8 or 9 digits.
The 99-quotient of the PWIs (Papadakis-Webster binary string (PWBS)) are also
listed. The last column lists the binary strings that are not PWBSs. 
}
\end{table}

\large

\section{Integers in a limiting 71-cycle}

\begin{table}[H]  \tiny
\begin{center}
An example of a 71-cycle attractor \\
\begin{tabular}{c|rr}
\hline
i & 71-cycle integers & divided by 99\\
\hline
1 &          9999010009899999 & (not PW)    101000101110101 \\ 
2 &         19998999010009998 &             202010091010202 \\ 
3 &        109989000109999989 &     (PW)    1111000001111111 \\ 
4 &       1099988901110989890 &           11110999001121110 \\ 
5 &        110098790012089989 &            1112108990021111 \\ 
6 &       1100079000109980000 &           11111909092020000 \\ 
7 &       1099179990100279989 &           11102828182831111 \\ 
8 &      10998900001099999890 &     (PW)  111100000011111110 \\ 
9 &       1098900991099009989 &  (not PW) 11100010011101111 \\ 
10 &      10997910893089108890 &          111090009021102110 \\ 
11 &       1117712853287128989 &           11290028821082111 \\ 
12 &      11015930676869306100 &          111272027039083900 \\ 
13 &      10855533809265355089 &          109651856659246011 \\ 
14 &     108910890100098910890 &         1100110001011100110 \\ 
15 &      10891000099000891089 &          110010102010110011 \\ 
16 &     108910800198000910890 &         1100109092909100110 \\ 
17 &      10891799306992891089 &          110018174818110011 \\ 
18 &     108911629267392610890 &         1100117467347400110 \\ 
19 &      10895335504466491089 &          110053893984510011 \\ 
20 &     108914801945019850890 &         1100149514596160110 \\ 
21 &      10855891395911431089 &          109655468645570011 \\ 
22 &     108869303355231286890 &         1099689932881124110 \\ 
23 &      10187170801927318089 &          102900715170983011 \\ 
24 &     108268543712734496190 &         1093621653663984810 \\ 
25 &      16574106495388633389 &          167415217125137711 \\ 
26 &     114907794854848780950 &         1160684796513624050 \\ 
27 &      55819946396351071539 &          563837842387384561 \\ 
28 &     149336961765716063394 &         1508454159249657206 \\ 
29 &     642697579332885697335 &         6491894740736219165 \\ 
30 &     108900991098909901089 &   (PW)  1100010011100100011 \\ 
31 &    1089010900989108910890 &   (PW)  11000110111001100110 \\ 
32 &     108812881099018801089 &         1099120011101200011 \\ 
33 &    1088921692089207019890 &        10999209011002091110 \\ 
34 &      99814662286245721089 &         1008228911982280011 \\ 
35 &       1801908018019079190 &           18201091091101810 \\ 
36 &        882198909910988109 &            8911100100110991 \\ 
37 &       1784087929820879397 &           18021090200210903 \\ 
38 &       9723868219118684268 &           98220891102208932 \\ 
39 &       1099000099990000989 &    (PW)   11101011111010111 \\ 
40 &      10989001099890010890 &   (PW)    111000011110000110 \\ 
41 &       1187991200879911989 &           11999911119999111 \\ 
42 &      11079190980901909800 &          111911020009110200 \\ 
43 &      10188280071992712789 &          102911919919118311 \\ 
44 &     108910009989001000890 &   (PW)   1100101111000010110 \\ 
45 &      10909908999100981089 &          110201101001020011 \\ 
46 &     108928809199081971990 &         1100291002010929010 \\ 
47 &       9749628207173142189 &           98481093001748911 \\ 
48 &      19562041924201411668 &          197596383072741532 \\ 
49 &     106173452167115438259 &         1072459112799145841 \\ 
50 &    1059007963928369809860 &        10697050140690604140 \\ 
51 &     369918325634672800359 &         3736548743784573741 \\ 
52 &    1322926602071196620322 &        13362894970416127478 \\ 
53 &    3553193513773262912553 &        35890843573467302147 \\ 
54 &       1000890000108999000 & (not PW)  10110000001101000 \\ 
55 &        990891990108018999 &           10009010001091101 \\ 
56 &       1990702791207217098 &           20108109002093102 \\ 
57 &      10897829813179288089 &          110079089022013011 \\ 
58 &     108986126945072167890 &         1100869969142143110 \\ 
59 &      10224856395450478089 &          103281377731823011 \\ 
60 &     108312261854816320290 &         1094063251058750710 \\ 
61 &      16288643396654106489 &          164531751481354611 \\ 
62 &     114748789065988794750 &         1159078677434230250 \\ 
63 &      57250899505000947339 &          578291914191928761 \\ 
64 &     150625799555600752614 &         1521472722783845986 \\ 
65 &     566882806111598278665 &         5726088950622204835 \\ 
66 &          9910999989990000 &  (not PW)   100111111010000 \\ 
67 &          9910000089999801 &             100101011010099 \\ 
68 &          8820000289999602 &              89090912020198 \\ 
69 &          6750000469999314 &              68181822929286 \\ 
70 &          2610000829998738 &              26363644747462 \\ 
71 &         10989000109998900 &  (PW)       111000001111100 \\ 
\hline
72=1 &        9999010009899999 &  (not PW)   101000101110101 \\ 
\hline
\end{tabular}
\end{center}
\caption{ \label{tableA2}
Elements in a 71-cycle of IDR. The 99-quotients
of the integers on the left column is listed in the right column.
If a 99-quotient is a binary sequence, we further checked if it
is a PWBS or not by the hairpin pairing rule (Theorem \ref{theorem-hairpin} 
and Fig.\ref{fig1}).
}
\end{table}

\large

\section{IDR beyond the decimal system}

\indent

A base-b length-n integer is defined as $D= \sum_{i=0}^n a_i b^i=(a_0a_1a_2 \cdots a_n)$,
where $b > 1$ is a positive integer, and $a_i \in (0,1,2, \cdots, b-2, b-1)$. The digital
reverse of $D$ is defined as before: $rev(D)=(a_n a_{n-1} \cdots a_2a_1a_0)$.
The Papadakis-Webster integers (PWI) in base-b system are defined the same:
E $\equiv$ D-rev(D), PWI $\equiv$  E+rev(E).
The mapping of Eq.\ref{eq-map} (IDR) is also defined the same as before.
Here, we describe some results concerning 1089 trick and IDR, 
generalized from decimal system to any base-b systems, omitting proofs.
We restrict ourselves to systems with $b > 2$, as the 
binary system ($b=2$) presents certain specific peculiarities.

\begin{itemize}
\item
The concept of divisibility for base-b integers is based on nodulo arithmetic
\citep{uspensky}.
Theorem \ref{theorem-99} that PWI is divisible by 99 for decimal system
is generalized to base-b system as: base-b PWI is divisible by (b-1)(b+1),
and the quotient is a binary string. Therefore, the concept of
Papadakis-Webster binary string (PWBS) remains to be true.

\item
Proposition \ref{prop-tenD} and corollary \ref{cor-palin} is generalized to:
for base-b integers in a limiting 2-cycle of IDR, 
$D_{I+1}= b D_I$, and $rev(D_I)/D_I=b-1$. 

\item
The propositions \ref{prop-v1} and \ref{prop-v1-only}, concerning 10(9)$_L$89
($L \ge 0$) motifs and its expansions as 2-cycle elements for IDR, 
can be generalized to 10[b-1]$_L$[b-2][b-1] ($L \ge 0$, where [b-1] and [b-2] 
are the symbols presenting values b-1 and b-2) motifs, as
2-cycle elements for base-b integers.
For example, for base-8 integers, the foundamental motif is 10(7)$_L$67;
for base-13 integers (where the digits are 0,1,2, $\cdots$, 9, a,b,c),
the motif is 10(c)$_L$bc.

\item
For limiting $p$-cycle of IDR, same cycle length may also appear
in other base-b systems with a similar pattern. See Table \ref{table-A3},
for example of a  $p=12$ cycle in base-8 (octal) and base-9 (nonal) integers,
as compared to the correponding decimal system.

\item
There are also $p$-cycles in base-b integer systems that do not have
a correspondence in the decimal system. Table \ref{table-A4} shows
such an example of a 41-cycle in base-4 (quaternary) system.

\item
Similar to Conjecture \ref{con-8} that diverging trajectories of IDR tend to have a
8-cycle rhythm in the middle section for the decimal system, there is
also a ubiquitous 8-cycle rhythm in the middle digits in base-b numerical systems:
[b-1][b-2], 0[b-2], [b-2], [b-1], [b-1][b-1], [b-1][b-1], [b-1],and [b-1].

\end{itemize}

\begin{table}[H]  
\begin{center}
An example of 12-cycle of IDR in base-8 and base-9 integers as compared to the decimal system\\
\begin{tabular}{c|rr|r}
\hline
i & b=8 & b=9 & b=10\\
\hline
1 &          77077   & 88088 & 99099 \\
2 & 176176      & 187187 & 198198 \\
3 & 1070067      & 1080078 & 1090089  \\   
4 & 10670770      & 10780880     & 10890990 \\    
5 & 761167      & 871178     & 981189  \\   
6 & 1742356      & 1852367     & 1962378  \\   
7 & 10475047      & 10585058     & 10695069   \\ 
8 & 104554450      & 105654560     & 106754670 \\    
9 & 30077047      & 30187058     & 30297069 \\    
10 & 124176052      & 125276162     & 126376272 \\    
11 & 375067473      & 387058683     & 399049893 \\    
12 & 106700      & 107800     & 108900  \\  
\hline
13=1 &          77077   & 88088 & 99099 \\
\hline
\end{tabular}
\end{center}
\caption{ \label{table-A3}
An IDR limiting 12-cycle in base-8 (octal) and base-9 (nonal)
numerical system that has a correspondence in the base-10 (decimal) system.
}
\end{table}

\begin{table}[H]  \small
\begin{center}
An example of 41-cycle of IDR in the base-4 system\\
\begin{tabular}{c|rr|r}
\hline
i & 41-cycle elements (b=4)\\
\hline
1 & 13333032313332  \\
2 & 103331022013323   \\
3 & 1033301302213230  \\
4 & 110113211113323  \\
5 & 1100030330031000  \\
8 & 1032123333130323  \\
7 & 10323103333003230 \\
8 & 1033010002210323  \\
9 & 10323132002320230  \\
10 & 1120211313121323  \\
11 & 11012031110302200  \\
12 & 10131123331221123  \\
13 & 102310003330000230  \\
14 & 10303310023321023  \\
15 & 102322302031311330  \\
16 &  3203111222022123 \\
17 & 13021320103201212 \\
18 & 100232210212113303 \\
19 & 1010210022231011310 \\
20 & 213102100030231203 \\
21 & 1121300130032033121 \\
22 &  3001203031002130332 \\
23 &  10231023033103323 \\
24 & 103221222131123130 \\
25 &  11300030303000223 \\
26 &  110100121212001200 \\
27 &  101333303031000123 \\
28 &  1023000100000333230 \\
29 &  33010033330330023 \\
30 &  1000123331322330 \\
31 & 1231332112323 \\
32 & 11130111110310 \\
33 &  3223000001133 \\
34 &  13200000011022 \\
35 & 101211000011313 \\
36 & 1020321000130020 \\
37 & 220010332233213 \\
38 & 1133003231303301 \\
39 & 33311302233330 \\
40 & 23312021321331 \\
41 & 3333103233333 \\
\hline
42=1 & 13333032313332  \\
\hline
\end{tabular}
\end{center}
\caption{ \label{table-A4}
A base-4 (quaternary) limiting 41-cycle of IDR that does not have a correspondence in decimal system.
}
\end{table}

\end{document}